\documentclass[conference]{IEEEtran}
\usepackage{amsmath}
\usepackage{amssymb}
\usepackage{amsfonts}
\usepackage{graphicx}
\usepackage{epsfig}
\usepackage{subfigure}
\usepackage{psfrag}
\usepackage{xcolor}

\begin{document}

\title{Secrecy Wireless Information and Power Transfer with MISO Beamforming}

\author{ \IEEEauthorblockN{Liang Liu, Rui
Zhang, and Kee-Chaing Chua}
\IEEEauthorblockA{ECE Department, National University
of Singapore. E-mails:\{liu\_liang,elezhang,eleckc\}@nus.edu.sg}}

\maketitle

\begin{abstract}

The dual use of radio signals for simultaneous wireless information and power transfer (SWIPT) has recently drawn significant attention. To meet the practical requirement that energy receivers (ERs) operate with much higher received power than information receivers (IRs), ERs need to be deployed closer to the transmitter than IRs. However, due to the broadcast nature of wireless channels, one critical issue is that the messages sent to IRs cannot be eavesdropped by ERs, which possess better channels from the transmitter. In this paper, we address this new secrecy communication problem in a multiuser multiple-input single-output (MISO) SWIPT system where a multi-antenna transmitter sends information and energy simultaneously to one IR and multiple ERs, each with a single antenna. By optimizing transmit beamforming vectors and their power allocation, we maximize the weighted sum-energy transferred to ERs subject to a secrecy rate constraint for the information sent to the IR. We solve this non-convex problem optimally by reformulating it into a two-stage problem. First, we fix the signal-to-interference-plus-noise ratio (SINR) at the IR and obtain the optimal beamforming solution by applying the technique of semidefinite relaxation (SDR). Then the original problem is solved by a one-dimension search over the optimal SINR value for the IR. Furthermore, two suboptimal low-complexity beamforming schemes are proposed, and their achievable (secrecy) rate-energy (R-E) regions are compared against that by the optimal scheme.

\end{abstract}

\IEEEpeerreviewmaketitle

\newtheorem{definition}{\underline{Definition}}[section]
\newtheorem{fact}{Fact}
\newtheorem{assumption}{Assumption}
\newtheorem{theorem}{\underline{Theorem}}[section]
\newtheorem{lemma}{\underline{Lemma}}[section]
\newtheorem{corollary}{Corollary}[section]
\newtheorem{proposition}{\underline{Proposition}}[section]
\newtheorem{example}{\underline{Example}}[section]
\newtheorem{remark}{\underline{Remark}}[section]
\newtheorem{algorithm}{\underline{Algorithm}}
\newcommand{\mv}[1]{\mbox{\boldmath{$ #1 $}}}
\newcommand{\smv}[1]{\small{\mbox{\boldmath{$ #1 $}}}}

\section{Introduction}\label{eqn:Introduction}

Recently, there has been an upsurge of interest in radio signal enabled simultaneous wireless information and power transfer (SWIPT) \cite{Sahai10}-\cite{Rui13}. A typical SWIPT system of practical interest is shown in Fig. \ref{fig1}, where a fixed access point (AP) with constant power supply broadcasts wireless signals to a set of user terminals (UTs), among of which some decode information from the received signals, thus referred to as information receivers (IRs), while the others harvest energy, thus called energy receivers (ERs) \cite{Rui11}. To meet the practical requirement that IRs and ERs typically operate with very different power sensitivity (e.g., $-60$dBm for IRs versus $-10$dBm for ERs), a receiver-location-based scheduling for information and energy transmission has been proposed in \cite{Rui11}, \cite{Liang13}, where ERs need to be in more proximity to the transmitter than IRs. However, this gives rise to a new information security issue since ERs, which have better channels than IRs from the transmitter, can easily eavesdrop the information sent to IRs. To tackle this challenging problem, in this paper we propose the use of multiple antennas at the transmitter to achieve the secret information transmission to IRs and yet maximize the energy simultaneously transferred to ERs, by properly designing the beamforming vectors and their power allocation at the transmitter. For the purpose of exposition, we consider a multiple-input single-output (MISO) SWIPT system with a multi-antenna transmitter, one single-antenna IR, and $K\geq 1$ single-antenna ERs, as shown in Fig. \ref{fig1}. To prevent the ERs from eavesdropping the information sent to the IR, we study the joint information and energy transmit beamforming design to maximize the weighted sum-energy transferred to ERs subject to a secrecy rate constraint for the IR. This problem is shown to be non-convex, but we solve it globally optimally by reformulating it into a two-stage optimization problem. First, we fix the signal-to-interference-plus-noise ratio (SINR) at the IR and obtain the optimal transmit beamforming solution by applying the semidefinite relaxation (SDR) technique. Next, the original problem is solved by a one-dimension search for the optimal SINR value at the IR. Furthermore, we present two suboptimal designs of lower complexity, for which the directions of information and energy beams are separately optimized with their power allocation. Finally, we compare the performance of the proposed optimal algorithm with that of the two suboptimal schemes by simulations.
\begin{figure}
\setlength{\abovecaptionskip}{-2pt}
\setlength{\belowcaptionskip}{-5pt}
\begin{center}
 \scalebox{0.35}{\includegraphics*{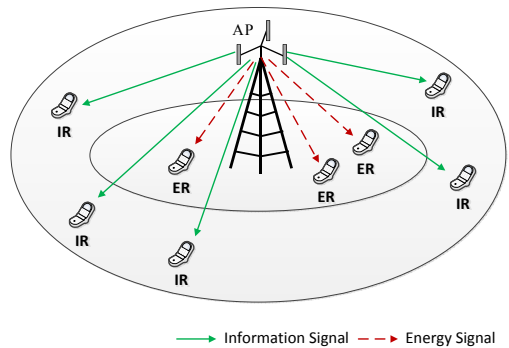}}
 \end{center}
\caption{A simultaneous wireless information and power transfer (SWIPT) system with ``near'' energy receivers (ERs) and ``far'' information receivers (IRs).} \label{fig4}\vspace{-17pt}
\end{figure}

It is worth noting that in \cite{Rui13}, a similar MISO SWIPT system as in this paper has been studied, but without the secret information transmission. As a result, it was shown in \cite{Rui13} that to maximize the weighted sum-energy transferred to ERs while meeting given SINR constraints at the IRs, the optimal strategy is to employ only information beams without any dedicated energy beam. However, in this paper we will show that with the newly introduced secrecy rate constraint at the IR, energy beams are in general needed in the optimal solution. The reason is that energy beams in this paper will also play an important role of artificial noise (AN) \cite{Goel06} to facilitate the secret information transmission to the IR by interfering with the ERs that may eavesdrop the IR's message. It is also worth noting that AN-aided secrecy communication has been extensively studied in the literature, where a fraction of the transmit power is allocated to send artificially generated noise signals to reduce the information eavesdropped by the eavesdroppers. Since in practice, eavesdroppers' channels are in general unknown at the transmitter, an isotropic approach was proposed in \cite{Goel06} where the power of AN is uniformly distributed in the null space of the legitimate receiver's channel, while the performance of this practical approach has been shown to be nearly optimal at the high signal-to-noise ratio (SNR) regime in \cite{Khisti07}. Furthermore, the MISO beamforming design problem for the AN-aided secrecy transmission under the assumption that eavesdroppers' channels are known at the transmitter has been studied in \cite{Ma11}. Notice that this assumption is not practically valid if the eavesdroppers are passive devices. However, for the SWIPT system of our interest, since ERs need to feed back their individual channels to the transmitter for it to deliver the required energy, it is reasonable to assume that ERs' channels are known at the transmitter.

The rest of this paper is organized as follows. Section \ref{sec:system model} presents the MISO SWIPT system model. Section \ref{sec:Problem Formulation} formulates the weighted sum-energy maximization problem subject to the secrecy rate constraint. Section \ref{sec:Optimal Solution} derives the optimal beamforming solution to this problem. Section \ref{sec:Suboptimal Solution} presents two suboptimal algorithms with lower complexity. Section \ref{sec:Numerical Results} provides numerical results on the performance of the proposed schemes. Finally, Section \ref{sec:Conclusion} concludes the paper.

\section{System Model}\label{sec:system model}

\begin{figure}
\setlength{\abovecaptionskip}{-5pt}
\setlength{\belowcaptionskip}{-8pt}
\begin{center}
 \scalebox{0.35}{\includegraphics*{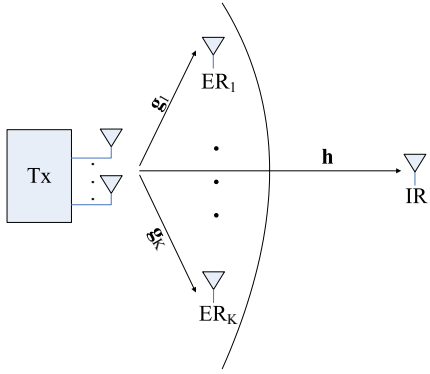}}
 \end{center}
\caption{A MISO SWIPT system with $K$ ERs and one single IR.} \label{fig1}\vspace{-10pt}
\end{figure}

We consider a multiuser MISO downlink system for SWIPT over a given frequency band as shown in Fig. \ref{fig1}. It is assumed that there is one single IR and $K$ ERs denoted by the set $\mathcal{K}_{{\rm EH}}=\{{\rm ER}_1,\cdots,{\rm ER}_K\}$, where the IR is more distant away from the transmitter (Tx) than all ERs to meet their different received power requirements. Suppose that Tx is equipped with $M>1$ antennas, while each IR/ER is equipped with one single antenna. We assume linear transmit beamforming at Tx and the IR is assigned with one dedicated information beam, while the $K$ ERs are in total allocated to $d\leq M$ energy beams without loss of generality. Therefore, the baseband transmitted signal of Tx can be expressed as
\begin{align}\label{eqn:signal3}
\mv{x}=\mv{v}_0s_0+\sum\limits_{i=1}^d\mv{w}_i s_i,
\end{align}where $\mv{v}_0\in \mathbb{C}^{M\times 1}$ and $\mv{w}_i\in \mathbb{C}^{M\times 1}$ denote the information beamforming vector and the $i$th energy beamforming vector, $1\leq i \leq d$, respectively; $s_0$ denotes the transmitted signal for the IR, while $s_i$'s, $i=1,\cdots,d$, denote the energy-carrying signals for energy beams. It is assumed that $s_0$ is a circularly symmetric complex Gaussian (CSCG) random variable with zero mean and unit variance, denoted by $s_0\sim \mathcal{CN}(0,1)$. Furthermore, $s_i$'s, $1\leq i \leq d$, are in general any arbitrary random signals each with unit average power. Since in this paper we consider secret information transmission to the IR, $s_i$'s, $1\leq i \leq d$, also play the role of AN to reduce the information rate eavesdropped by the ERs \cite{Goel06}. As a result, similarly to \cite{Goel06}-\cite{Ma11}, we assume that $s_i$'s are independent and identically distributed (i.i.d.) CSCG random variables denoted by $s_i\sim \mathcal{CN}(0,1)$, $\forall i$, since the worst-case noise for the eavesdropping ERs is known to be Gaussian. Suppose that Tx has a transmit sum-power constraint $\bar{P}$; from (\ref{eqn:signal3}), we thus have $E[\mv{x}^H\mv{x}]=\|\mv{v}_0\|^2+\sum_{i=1}^d\|\mv{w}_i\|^2\leq \bar{P}$.

In this paper, we assume a quasi-static fading environment and for convenience denote $\mv{h}\in \mathbb{C}^{M\times 1}$ and $\mv{g}_k\in \mathbb{C}^{M\times 1}$ as the conjugated complex channel vectors from Tx to IR and ${\rm ER}_k$, $k=1,\cdots,K$, respectively, where $\mv{h}$ and $\mv{g}_k$'s are assumed to be independently CSCG distributed with zero mean and covariance matrices $\rho_h^2\mv{I}$ and $\rho_{g_k}^2\mv{I}$, respectively, with $\rho_{g_k}^2> \rho_h^2$, $\forall k$. It is further assumed that $\mv{h}$ and $\mv{g}_k$'s are perfectly known at Tx. The signal received at IR is then expressed as
\begin{align}\label{eqn:signal1}
y_0=\mv{h}^H\mv{x}+z_0,
\end{align}where $z_0\sim \mathcal{CN}(0,\sigma_0^2)$ denotes the additive noise at IR. Furthermore, the signal received at ${\rm ER}_k$ can be expressed as
\begin{align}\label{eqn:signal2}
y_k=\mv{g}_k^H\mv{x}+z_k, ~~~ k=1,\cdots,K,
\end{align}where $z_k\sim \mathcal{CN}(0,\sigma_k^2)$ denotes the additive noise at ${\rm ER}_k$. It is assumed that $z_k$'s are independent over $k$.

According to (\ref{eqn:signal1}), the SINR at IR can be expressed as
\begin{align}\label{eqn:IR SINR}
{\rm SINR}_0=\frac{|\mv{v}_0^H\mv{h}|^2}{\sum\limits_{i=1}^d|\mv{w}_i^H\mv{h}|^2+\sigma_0^2},
\end{align}while according to (\ref{eqn:signal2}), the SINR at ${\rm ER}_k$ (suppose that it is an eavesdropper to decode the message for the IR instead of harvesting energy) can be expressed as
\begin{align}\label{eqn:ER SINR}
{\rm SINR}_k=\frac{|\mv{v}_0^H\mv{g}_k|^2}{\sum\limits_{i=1}^d|\mv{w}_i^H\mv{g}_k|^2+\sigma_k^2}, ~~~ k=1,\cdots,K.
\end{align}The achievable secrecy rate at IR is thus given by \cite{Poor07}:
\begin{align}\label{eqn:secrecy rate}
r_0=\min\limits_{1\leq k\leq K} ~~~ \log_2\left(1+{\rm SINR}_0\right)- \log_2\left(1+{\rm SINR}_k \right).
\end{align}

On the other hand, for wireless power transfer, due to the broadcast nature of wireless channels, the energy carried by all information and energy beams, i.e., $\mv{v}_0$ and $\mv{w}_i$'s ($1\leq i \leq d$), can all be harvested at each ${\rm ER}_k$. Hence, assuming unit slot duration, the harvested energy of ${\rm ER}_k$ in each slot is given by \cite{Rui11}:
\begin{align}\label{eqn:harvested energy}
E_k=\zeta \left(|\mv{v}_0^H\mv{g}_k|^2+\sum\limits_{i=1}^d|\mv{w}_i^H\mv{g}_k|^2\right), ~~~ 1\leq k \leq K,
\end{align}where $0<\zeta<1$ denotes the energy harvesting efficiency.

\section{Problem Formulation}\label{sec:Problem Formulation}
In this paper, we aim to maximize the weighted sum-energy transferred to all $K$ ERs subject to the secrecy rate constraint at IR. Therefore, the following problem is formulated.
\begin{align*}\mathrm{(P1)}:~\mathop{\mathtt{Maximize}}_{\mv{v}_0,\{\mv{w}_i\}}
& ~~~  \sum\limits_{k=1}^K \mu_k\zeta \left(|\mv{v}_0^H\mv{g}_k|^2+\sum\limits_{i=1}^d|\mv{w}_i^H\mv{g}_k|^2\right)  \\
\mathtt {Subject \ to} & ~~~ r_0 \geq \bar{r}_0 , \\ & ~~~
\|\mv{v}_0\|^2+\sum\limits_{i=1}^d\|\mv{w}_i\|^2 \leq \bar{P},
\end{align*}where $\mu_k\geq 0$ denotes the non-negative energy weight for ${\rm ER}_k$, and $\bar{r}_0$ denotes the target secrecy rate for IR.

Note that there are two conflicting goals in designing the information beamforming vector $\mv{v}_0$ for Problem (P1). On one hand, to maximize the weighted sum-energy, the power of the received signal at ${\rm ER}_k$ due to the information beam, i.e., $|\mv{v}_0^H\mv{g}_k|^2$, is desired to be as large as possible. However, on the other hand, from the viewpoint of secrecy rate maximization according to (\ref{eqn:IR SINR}) and (\ref{eqn:ER SINR}), it follows that $|\mv{v}_0^H\mv{g}_k|^2$ should be minimized at ${\rm ER}_k$ to avoid any ``leakage'' information. To resolve this conflict, we need to resort to the energy beamforming vectors $\mv{w}_i$'s, since they not only provide dedicated wireless energy to ERs, but also play the important role of AN to reduce the eavesdropper's SINR in (\ref{eqn:ER SINR}) at each ER.

Before we proceed to solve (P1), we first study the feasibility of this problem for a given pair of $\bar{r}_0$ and $\bar{P}$. This reduces to the case where no energy transfer is required and $\mv{w}_i$'s play the only role of AN. In this case, the algorithm proposed in \cite{Ma11} can be applied to check the feasibility of (P1). Without loss of generality, in the rest of this paper, we assume that (P1) is feasible.

Next, we consider another special case of (P1) when there is no IR. In this case, $\mv{v}_0=\mv{0}$ and thus (P1) reduces to\begin{align*}\mathrm{({\rm P1-NoI})}:~\mathop{\mathtt{Maximize}}_{\{\mv{w}_i\}}
& ~~~  \sum\limits_{k=1}^K \mu_k\zeta \left(\sum\limits_{i=1}^d|\mv{w}_i^H\mv{g}_k|^2\right)  \\
\mathtt {Subject \ to} & ~~~
\sum\limits_{i=1}^d\|\mv{w}_i\|^2 \leq \bar{P}.
\end{align*}

Let $\psi$ and $\mv{\eta}$ denote the maximum eigenvalue and its corresponding unit-norm eigenvector of the matrix $\sum_{k=1}^K\mu_k\zeta\mv{g}_k\mv{g}_k^H$, respectively. From \cite{Rui11}, the optimal value of Problem (P1-NoI) is known to be\begin{align}\label{eqn:max energy}
E_{{\rm max}}=\psi\bar{P},\end{align}which is achieved by $\mv{w}_i^\ast=\sqrt{p_i}\mv{\eta}$, $1\leq i \leq d$, for any set of $p_i$'s satisfying $\sum_{i=1}^dp_i=\bar{P}$. In other words, the optimal solution to Problem (P1-NoI) is to align all the energy beams to the same direction as $\mv{\eta}$. Now suppose that the IR is added back to (P1). In order to achieve the same maximum weighted sum-energy $E_{{\rm max}}$, it is easy to show that $\mv{v}_0=\sqrt{p_0}\mv{\eta}$ and $\mv{w}_i=\sqrt{p_i}\mv{\eta}$, $1\leq i \leq d$, for any set of $p_0$, $p_1,\cdots,p_d$ with $\sum_{i=0}^dp_i= \bar{P}$, must be a feasible solution to Problem (P1). With this solution set, the maximum achievable secrecy rate (if non-negative) can be shown from (\ref{eqn:secrecy rate}) to be
\begin{small}
\begin{align}\label{eqn:sec rate}
\bar{R}=\min\limits_{1\leq k\leq K} \ \log_2\left(1+\frac{\bar{P}|\mv{\eta}^H\mv{h}|^2}{\sigma_0^2}\right) -  \log_2\left(1+\frac{\bar{P}|\mv{\eta}^H\mv{g}_k|^2}{\sigma_k^2}\right),
\end{align}\end{small}which is achieved by setting $p_0=\bar{P}$ and $p_i=0$, $1\leq i \leq d$. As a result, in Problem (P1), if $\bar{r}_0\leq \max(0,\bar{R})$, then $\mv{v}_0^\ast=\bar{P}\mv{\eta}$ and $\mv{w}_i^\ast=\mv{0}$, $1\leq i \leq d$, is the optimal solution since the maximum objective value $E_{{\rm max}}$ is achieved and the secrecy rate constraint for the IR, i.e., $\bar{r}_0$, is satisfied. Without loss of generality, in the rest of this paper we assume that $\bar{r}_0>\max(0,\bar{R})$ in Problem (P1) to remove the above trivial case of $\bar{r}_0\leq \max(0,\bar{R})$.

\section{Optimal Solution}\label{sec:Optimal Solution}

Note that (P1) is in general (i.e., with $\bar{r}_0>\max(0,\bar{R})$) a non-convex optimization problem since both the objective function and the secrecy rate constraint are non-concave functions with respect to $\mv{v}_0$ and $\mv{w}_i$'s. In this section, we propose an SDR-based algorithm to solve Problem (P1) optimally by reformulating it into a two-stage optimization problem. First, we have the following lemma.
\begin{lemma}\label{lemma1}
For Problem (P1), there exists a SINR value $\gamma_0>0$ such that the following problem
\begin{align*}\mathrm{(P1.1)}:~~~\mathop{\mathtt{Maximize}}_{\mv{v}_0,\{\mv{w}_i\}}
& ~~~ \sum\limits_{k=1}^K \mu_k\zeta \left(|\mv{v}_0^H\mv{g}_k|^2+\sum\limits_{i=1}^d|\mv{w}_i^H\mv{g}_k|^2\right) \\
\mathtt {Subject \ to} & ~~~ \frac{|\mv{v}_0^H\mv{h}|^2}{\sum\limits_{i=1}^d|\mv{w}_i^H\mv{h}|^2+\sigma_0^2}\geq \gamma_0, \\ & ~~~ \frac{|\mv{v}_0^H\mv{g}_k|^2}{\sum\limits_{i=1}^d|\mv{w}_i^H\mv{g}_k|^2+\sigma_k^2}\leq \frac{1+\gamma_0}{2^{\bar{r}_0}}-1, \ \forall k,  \\ & ~~~
\|\mv{v}_0\|^2+\sum\limits_{i=1}^d\|\mv{w}_i\|^2 \leq \bar{P},
\end{align*}has the same optimal solution to Problem (P1).
\end{lemma}

\begin{proof}
The proof follows easily by showing that for any given optimal solution of (P1), denoted by $\mv{v}_0^\ast$ and $\{\mv{w}_i^\ast\}$, it is also optimal for (P1.1) with $\gamma_0=|\mv{h}^H\mv{v}_0^\ast|^2/(\sum_{i=1}^d|\mv{h}^H\mv{w}_i^\ast|^2+\sigma_0^2)$.
\end{proof}

Let $g(\gamma_0)$ denote the optimal value of Problem (P1.1) with a given $\gamma_0>0$. Then, we have the following lemma.
\begin{lemma}\label{lemma2}
The optimal value of Problem (P1) is the same as that of the following problem
\begin{align*}\mathrm{(P1.2)}:~\mathop{\mathtt{Maximize}}_{\gamma_0>0}
& ~~~ g(\gamma_0).
\end{align*}

\end{lemma}

\begin{proof}
A sketch of the proof is provided here, while the details will be presented in the journal version of this paper \cite{Rui}. Let $t^\ast$ denote the optimal value of Problem (P1). From Lemma \ref{lemma1}, we have
\begin{align}\label{eqn:P1.1}\max\limits_{\gamma_0>0} \ g(\gamma_0)\geq t^\ast.\end{align}On the other hand, given any $\gamma_0>0$ for Problem (P1.1), it can be shown that any optimal solution $\mv{v}_0^\ast$ and $\{\mv{w}_i^\ast\}$ to Problem (P1.1) is feasible for Problem (P1). Therefore, we have
\begin{align}\label{eqn:p1}t^\ast \geq g(\gamma_0), \ \forall \gamma_0>0.\end{align}From (\ref{eqn:P1.1}) and (\ref{eqn:p1}), it follows that $\max\limits_{\gamma_0>0} \ g(\gamma_0) =t^\ast$. The lemma is thus proved.

\end{proof}

Let $\gamma_0^\ast$ denote the optimal solution to Problem (P1.2). Lemma \ref{lemma2} then implies that with $\gamma_0=\gamma_0^\ast$, Problems (P1) and (P1.1) have the same optimal solution. Therefore, Problem (P1) can be solved in the following two steps: First, given any $\gamma_0>0$, we solve Problem (P1.1) to find $g(\gamma_0)$; then, we solve Problem (P1.2) to find the optimal $\gamma^\ast_0$ by a one-dimension search over $\gamma_0>0$. Hence, in the rest of this section, we focus on solving (P1.1).

Note that (P1.1) is still non-convex. Define
$\mv{S}=\mv{v}_0\mv{v}_0^H$ and $\mv{Q}=\sum_{i=1}^d\mv{w}_i\mv{w}_i^H$. Then it follows that ${\rm rank}(\mv{S})\leq 1$ and ${\rm rank}(\mv{Q})\leq d$. By ignoring the above rank-one constraint on $\mv{S}$, the SDR of Problem (P1.1) can be expressed as
\begin{align}\mathrm{(P1.1-SDR)}:\nonumber \\ \mathop{\mathtt{Maximize}}_{\mv{S},\mv{Q}}
& ~~~ \sum\limits_{k=1}^K \mu_k\zeta\left({\rm Tr}(\mv{G}_k\mv{S})+{\rm Tr}(\mv{G}_k\mv{Q})\right) \nonumber \\
\mathtt {Subject \ to} & ~~~ {\rm Tr}(\mv{H}\mv{S})\geq \gamma_0 \left({\rm Tr}(\mv{H}\mv{Q})+\sigma_0^2\right), \label{eqn:constraint1}\\ & ~~~ \frac{{\rm Tr}(\mv{G}_k\mv{S})}{\gamma_e} \leq {\rm Tr}(\mv{G}_k\mv{Q})+\sigma_k^2, \ \forall k,  \label{eqn:constraint2}\\ & ~~~
{\rm Tr}(\mv{S})+{\rm Tr}(\mv{Q}) \leq \bar{P}, \label{eqn:constraint3}\\ & ~~~ \mv{S}\succeq \mv{0}, ~~~ \mv{Q}\succeq \mv{0},
\end{align}where $\mv{H}=\mv{h}\mv{h}^H$, $\mv{G}_k=\mv{g}_k\mv{g}_k^H$, and $\gamma_e=(1+\gamma_0)/2^{\bar{r}_0}-1$. If the optimal solution to Problem (P1.1-SDR), denoted by $\mv{S}^\ast$ and $\mv{Q}^\ast$, satisfies ${\rm rank}(\mv{S}^\ast)= 1$, then the optimal information beam $\mv{v}_0^\ast$ and $d$ energy beams $\mv{w}_i^\ast$'s, $i=1,\cdots,d$, to Problem (P1.1) can be obtained from the eigenvalue decompositions (EVDs) of $\mv{S}^\ast$ and $\mv{Q}^\ast$, respectively; otherwise, if ${\rm rank}(\mv{S}^\ast)>1$, the optimal value of Problem (P1.1-SDR) only serves as an upper bound on that of Problem (P1.1). In the following, we check whether ${\rm rank}(\mv{S}^\ast)=1$ holds or not for (P1.1-SDR).

Since (P1.1-SDR) is convex and satisfies the Slater's condition, its duality gap is zero \cite{Boyd04}. Let $\lambda$, $\beta_k$, and $\theta$ denote the dual variables associated with the SINR constraint of IR, the SINR constraint of ${\rm ER}_k$, $k=1,\cdots,K$, and the sum-power constraint in Problem (P1.1-SDR), respectively. Then the Lagrangian of Problem (P1.1-SDR) is expressed as
\begin{align}\label{eqn:Lagrangian}
L(\mv{S},\mv{Q},\lambda,\{\beta_k\},\theta)= & {\rm Tr}(\mv{A}\mv{S})+{\rm Tr}(\mv{B}\mv{Q}) \nonumber \\ & -\lambda \gamma_0 \sigma_0^2+\sum\limits_{k=1}^K\beta_k\gamma_e\sigma_k^2+\theta \bar{P},
\end{align}where
\begin{align}
& \mv{A}=\sum\limits_{k=1}^K\mu_k\zeta\mv{G}_k+\lambda\mv{H}-\sum\limits_{k=1}^K\beta_k\mv{G}_k-\theta\mv{I}, \label{eqn:A} \\
& \mv{B}=\sum\limits_{k=1}^K\mu_k\zeta\mv{G}_k-\lambda \gamma_0\mv{H}+\sum\limits_{k=1}^K\beta_k \gamma_e \mv{G}_k-\theta \mv{I}. \label{eqn:B}
\end{align}Let $\lambda^\ast\geq 0$, $\{\beta_k^\ast\geq 0\}$ and $\theta^\ast \geq 0$ denote the optimal dual solution to Problem (P1.1-SDR). Then, we have the following lemma.

\begin{lemma}\label{lemma3}
Given $\bar{r}_0>\max(0,\bar{R})$, the optimal dual solution to Problem (P1.1-SDR) satisfies that $\lambda^\ast>0$ and $\theta^\ast>0$.
\end{lemma}
\begin{proof}
Please refer to the longer version of this paper \cite{Rui}.
\end{proof}

Lemma \ref{lemma3} implies that for the optimal solution of Problem (P1.1-SDR), the constraints in (\ref{eqn:constraint1}) and (\ref{eqn:constraint3}) must be satisfied with equality due to the complementary slackness \cite{Boyd04}. Define $\mv{D}^\ast=\sum_{k=1}^K\mu_k\zeta \mv{G}_k-\lambda^\ast\gamma_0\mv{H}-\sum_{k=1}^K\beta_k^\ast\mv{G}_k-\theta^\ast\mv{I}$, and $l={\rm rank}(\mv{D}^\ast)$. Let $\mv{\Pi}\in \mathbb{C}^{M\times (M-l)}$ denote the orthogonal basis of the null space of $\mv{D}^\ast$, where $\mv{\Pi}=\mv{0}$ if $l=M$, and $\mv{\pi}_n$ denote the $n$th column of $\mv{\Pi}$. Then, based on Lemma \ref{lemma3}, we have the following proposition.

\begin{proposition}\label{proposition1}
The optimal solution $(\mv{S}^\ast,\mv{Q}^\ast)$ to Problem (P1.1-SDR) satisfies the following conditions:
\begin{itemize}
\item[1.] ${\rm rank}(\mv{Q}^\ast)\leq \min(K,M)$.
\item[2.] $\mv{S}^\ast$ can be expressed as\begin{align}\label{eqn:feasible rank1}\mv{S}^\ast=\sum\limits_{n=1}^{M-l}a_n\mv{\pi}_n\mv{\pi}_n^H+b\mv{\tau}\mv{\tau}^H,\end{align}where $a_n\geq 0$, $\forall n$, $b>0$, and $\mv{\tau}\in \mathbb{C}^{M\times 1}$ has unit-norm and satisfies $\mv{\tau}^H\mv{\Pi}=\mv{0}$.
\item[3.] If $\mv{S}^\ast$ given in (\ref{eqn:feasible rank1}) has the rank larger than one, i.e., there at least exists an $n$ such that $a_n>0$, then the following solution
\begin{align}
& \bar{\mv{S}}^\ast=b\mv{\tau}\mv{\tau}^H, \label{eqn:new S1}\\
& \bar{\mv{Q}}^\ast=\mv{Q}^\ast+\sum\limits_{n=1}^{M-l}a_n\mv{\pi}_n\mv{\pi}_n^H, \label{eqn:new Q1}
\end{align}is also optimal to Problem (P2.1-SDR), with ${\rm rank}(\bar{\mv{S}}^\ast)=1$.
\end{itemize}
\end{proposition}

\begin{proof}
Please refer to the longer version of this paper \cite{Rui}.
\end{proof}

With Proposition \ref{proposition1}, we are ready to find the optimal solution to problem (P1.1-SDR) with a rank-one covariance matrix for $\mv{S}$ as follows. First, we solve (P1.1-SDR) via CVX. If the obtained solution $(\mv{S}^\ast,\mv{Q}^\ast)$ satisfies ${\rm rank}(\mv{S}^\ast)=1$, then the solution is completed. Otherwise, if ${\rm rank}(\mv{S}^\ast)>1$, we can construct a new solution $(\bar{\mv{S}}^\ast,\bar{\mv{Q}}^\ast,\bar{t}^\ast)$ with ${\rm rank}(\bar{\mv{S}}^\ast)=1$ according to (\ref{eqn:feasible rank1})-(\ref{eqn:new Q1}). Therefore, the rank-one relaxation on $\mv{S}$ in (P1.1-SDR) results in no loss of optimality to (P1.1), and given any $\gamma_0>0$, the value of $g(\gamma_0)$ can be obtained by solving (P1.1-SDR). Furthermore, since ${\rm rank}(\mv{Q}^\ast)\leq \min(K,M)$ in Proposition \ref{proposition1}, it implies that in the case of $K<M$, at most $K$ energy beams are needed in the optimal solution of (P1), i.e., $d\leq K$.

\section{Suboptimal Solutions}\label{sec:Suboptimal Solution}

The optimal solution to (P1) proposed in Section \ref{sec:Optimal Solution} requires a joint optimization of the information/energy beams and their power allocation. In this section, we propose two suboptimal solutions for (P1) which can be obtained with lower complexity. Similar to \cite{Goel06}, in both of our proposed suboptimal solutions, the energy beams $\mv{w}_i$'s are all restricted to lie in the null space of the IR's channel $\mv{h}$ such that they cause no interference to the IR. However, the information beam $\mv{v}_0$ is aligned to the null space of the ERs' channels $\mv{G}=[\mv{g}_1,\cdots,\mv{g}_K]^H$ in the first suboptimal scheme in order to eliminate the information leaked to ERs, but to the same direction as $\mv{h}$ in the second suboptimal scheme to maximize the IR's SINR. Note that the first suboptimal scheme is only applicable when $K<M$ since otherwise the null space of $\mv{G}$ is empty. In the following, we present the two suboptimal schemes with more details.

\subsubsection{Suboptimal Scheme \uppercase\expandafter{\romannumeral1}}

Supposing that $K<M$, then the first suboptimal scheme aims to solve Problem (P1) with the additional constraints $\mv{v}_0^H\mv{g}_k=0$, $\forall k$, and $\mv{w}_i^H\mv{h}=0$, $\forall i$. Consider first the information beam $\mv{v}_0$. Let the singular value decomposition (SVD) of $\mv{G}$ be denoted as
\begin{align}
\mv{G}=\mv{U}\mv{\Lambda}\mv{V}^H=\mv{U}\mv{\Lambda}[\bar{\mv{V}} \ \tilde{\mv{V}}]^H,
\end{align}where $\mv{U}\in \mathbb{C}^{K\times K}$ and $\mv{V}\in \mathbb{C}^{M\times M}$ are unitary matrices, i.e., $\mv{U}\mv{U}^H=\mv{U}^H\mv{U}=\mv{I}$, $\mv{V}\mv{V}^H=\mv{V}^H\mv{V}=\mv{I}$, and $\mv{\Lambda}$ is a $K\times M$ rectangular diagonal matrix. Furthermore, $\bar{\mv{V}}\in \mathbb{C}^{M\times K}$ and $\tilde{\mv{V}}\in \mathbb{C}^{M\times (M-K)}$ consist of the first $K$ and the last $M-K$ right singular vectors of $\mv{G}$, respectively. It can be shown that $\tilde{\mv{V}}$ with $\tilde{\mv{V}}^H\tilde{\mv{V}}=\mv{I}$ forms an orthogonal basis for the null space of $\mv{G}$. Thus, to guarantee that $\mv{G}\mv{v}_0=\mv{0}$, $\mv{v}_0$ must be in the following form:
\begin{align}
\mv{v}_0=\sqrt{\tilde{P}_0}\tilde{\mv{V}}\tilde{\mv{v}}_0,
\end{align}where $\tilde{P}_0=\|\mv{v}_0\|^2$ denotes the transmit power of the information beam, and $\tilde{\mv{v}}_0$ is an arbitrary $(M-K)\times 1$ complex vector of unit norm. It can be shown that to maximize the IR's SINR, $\tilde{\mv{v}}_0$ should be aligned to the same direction as the equivalent channel $\tilde{\mv{V}}^H\mv{h}$, i.e., $\tilde{\mv{v}}_0^\ast=\tilde{\mv{V}}^H\mv{h}/\|\tilde{\mv{V}}^H\mv{h}\|$. Since the energy beams are all aligned to the null space of $\mv{h}$ (to be shown later), the secrecy rate of the IR in this scheme is
\begin{align}
r_0^{({\rm I})}=\log_2\left(1+\frac{\tilde{P}_0\|\tilde{\mv{V}}^H\mv{h}\|^2}{\sigma_0^2}\right).
\end{align}Note that the ERs cannot harvest any energy from the information beam; thus, to maximize the weighted sum-energy transferred to the ERs, $\tilde{P}_0$ should be as small as possible, i.e., $r_0^{({\rm I})}=\bar{r}_0$. It thus follows that \begin{align}\tilde{P}_0^\ast=\frac{(2^{\bar{r}_0}-1)\sigma_0^2}{\|\tilde{\mv{V}}^H\mv{h}\|^2}.\end{align}To summarize, in this scheme, we have
\begin{align}\label{eqn:sub1 v0}
\mv{v}_0^\ast=\frac{\sqrt{(2^{\bar{r}_0}-1)\sigma_0^2}\tilde{\mv{V}}\tilde{\mv{V}}^H\mv{h}}{\|\tilde{\mv{V}}^H\mv{h}\|^2}.
\end{align}

Next, consider the energy beams $\mv{w}_i$'s. Define the projection matrix as $\mv{T}=\mv{I}-\mv{h}\mv{h}^H/\|\mv{h}\|^2$. Without loss of generality, we can express $\mv{T}=\tilde{\mv{X}}\tilde{\mv{X}}^H$, where $\tilde{\mv{X}}\in \mathbb{C}^{M\times (M-1)}$ satisfies $\tilde{\mv{X}}^H\tilde{\mv{X}}=\mv{I}$. It can be shown that $\tilde{\mv{X}}$ forms an orthogonal basis for the null space of $\mv{h}$. Thus, to guarantee that $\mv{h}^H\mv{w}_i=0$, $\mv{w}_i$ must be in the following form:
\begin{align}
\mv{w}_i=\tilde{\mv{X}}\tilde{\mv{w}}_i, ~~~ i=1,\cdots,d,
\end{align}where $\tilde{\mv{w}}_i$ is an arbitrary $(M-1)\times 1$ complex vector. In this case, the energy harvested at ${\rm ER}_k$ is thus $\zeta\sum_{i=1}^d|\mv{w}_i^H\mv{g}_k|^2=\zeta\left(\sum_{i=1}^d\tilde{\mv{w}}_i^H\tilde{\mv{G}}_k\tilde{\mv{w}}_i\right)$, where $\tilde{\mv{G}}_k=\tilde{\mv{X}}^H\mv{G}_k\tilde{\mv{X}}$. To find the optimal $\tilde{\mv{w}}_i^\ast$'s, we need to solve the following problem.
\begin{align*}\mathrm{(P1-Sub1)}:~\mathop{\mathtt{Maximize}}_{\{\tilde{\mv{w}}_i\}}
& ~~~  \sum\limits_{k=1}^K \mu_k\zeta \left(\sum\limits_{i=1}^d\tilde{\mv{w}}_i^H\tilde{\mv{G}}_k\tilde{\mv{w}}_i\right)  \\
\mathtt {Subject \ to} & ~~~
\sum\limits_{i=1}^d\|\tilde{\mv{w}}_i\|^2 \leq \bar{P}-\tilde{P}_0^\ast.
\end{align*}

Let $\tilde{\psi}$ and $\tilde{\mv{\eta}}$ denote the maximum eigenvalue and its corresponding unit-norm eigenvector of the matrix $\sum_{k=1}^K\mu_k\zeta\tilde{\mv{G}}_k$, respectively. Similar to Problem (P1-NoI), it can be shown that the optimal value of Problem (P1-Sub1) is $\tilde{E}_{{\rm max}}=\tilde{\psi}(\bar{P}-\tilde{P}_0^\ast)$, which is achieved by $\tilde{\mv{w}}_i^\ast=\sqrt{\tilde{p}_i}\tilde{\mv{\eta}}$, $1\leq i \leq d$, for any set of $\tilde{p}_i$'s satisfying $\sum_{i=1}^d\tilde{p}_i=\bar{P}-\tilde{P}_0^\ast$. In practice, it is preferred to send only one energy beam to minimize the complexity of beamforming implementation at the transmitter, i.e.,
\begin{align}\label{eqn:sub1 wi}
\mv{w}_i^\ast=\left\{\begin{array}{ll}\sqrt{\bar{P}-\tilde{P}_0^\ast}\tilde{\mv{X}}\tilde{\mv{\eta}}, & {\rm if} \ i=1, \\ 0, & {\rm otherwise}.\end{array}\right.
\end{align}

\subsubsection{Suboptimal Scheme \uppercase\expandafter{\romannumeral2}}

The second suboptimal scheme aims to solve Problem (P1) with the additional constraints $\mv{v}_0=\sqrt{\hat{P}_0}\mv{h}/\|\mv{h}\|$ and $\mv{w}_i^H\mv{h}=0$, $\forall i$, where $\hat{P}_0=\|\mv{v}_0\|^2$ denotes the transmit power of the information beam. Similar to (\ref{eqn:sub1 wi}), it can be shown that the optimal energy beams should be in the following form
\begin{align}
\mv{w}_i=\left\{\begin{array}{ll}\sqrt{\bar{P}-\hat{P}_0}\tilde{\mv{X}}\tilde{\mv{\eta}}, & {\rm if} \ i=1, \\ 0, & {\rm otherwise}.\end{array}\right.
\end{align}Next, we derive the optimal power allocation $\hat{P}_0^\ast$. It can be shown that the secrecy rate of the IR in this scheme is given by
\begin{align}
r_0^{({\rm II})}=& \min_{1\leq k \leq K} ~  \log_2\left(1+\frac{\hat{P}_0\|\mv{h}\|^2}{\sigma_0^2}\right)\nonumber \\ &-\log_2\left(1+\frac{\hat{P}_0|\mv{h}^H\mv{g}_k|^2}{\|\mv{h}\|^2((\bar{P}-\hat{P}_0)|\tilde{\mv{\eta}}^H\tilde{\mv{X}}^H\mv{g}_k|^2+\sigma_k^2)}\right).
\end{align}Define the set of feasible power allocation as $\hat{\mathcal{P}}_0=\{\hat{P}_0|r_0^{({\rm II})}\geq \bar{r}_0, 0<\hat{P}_0\leq \bar{P}\}$. Then let $\hat{P}_0^{{\rm min}}$ and $\hat{P}_0^{{\rm max}}$ denote the minimal and maximal elements in the set $\hat{\mathcal{P}}_0$, respectively. Thus, to maximize the weighted sum-energy transferred to ERs subject to the secrecy rate constraint of the IR, the optimal power allocation can be expressed as
\begin{align}
\hat{P}_0^\ast=\left\{\begin{array}{ll}\hat{P}_0^{{\rm max}}, & {\rm if} \ \frac{\sum\limits_{k=1}^K\mu_k|\mv{h}^H\mv{g}_k|^2}{\|\mv{h}\|^2}\geq \sum\limits_{k=1}^K\mu_k|\tilde{\mv{\eta}}^H\tilde{\mv{X}}^H\mv{g}_k|^2, \\ \hat{P}_0^{{\rm min}}, & {\rm otherwise}.\end{array}\right.
\end{align}

\section{Numerical Results}\label{sec:Numerical Results}

In this section, we provide numerical examples to validate our results. It is assumed that Tx is equipped with $M=4$ antennas, and there are $K=3$ ERs.\footnote{Suboptimal scheme \uppercase\expandafter{\romannumeral1} is thus applicable in this case.} We assume that the signal attenuation from Tx to all ERs is $30$dB corresponding to an identical distance of $1$ meter, i.e., $\rho_{g_k}^2=-30$dB, $1\leq k \leq K$, and that to the IR is $70$dB corresponding to a distance of $20$ meters, i.e., $\rho_h^2=-70$dB. The channel vectors $\mv{g}_k$'s and $\mv{h}$ are randomly generated from i.i.d. Rayleigh fading with the respective average power values specified as above. We set $\bar{P}=1$Watt (W) or $30$dBm, $\zeta=50\%$, and $\sigma_k^2=-50$dBm, $0\leq k \leq K$. We also set $\mu_k=1$, $1\leq k \leq K$; thus, the sum-energy harvested by all ERs is considered.

Similar to \cite{Liang13}, we use the Rate-Energy (R-E) region, which consists of all the achievable (secrecy) rate and harvested sum-energy pairs for a given sum-power constraint $\bar{P}$, to compare the performance of the optimal and suboptimal schemes proposed in Sections \ref{sec:Optimal Solution} and \ref{sec:Suboptimal Solution}. Specifically, the R-E region is defined as
\begin{align}\label{eqn:R-E region}\mathcal{C}_{{\rm R-E}}\triangleq \bigcup\limits_{\|\mv{v}_0\|^2+\sum\limits_{i=1}^d\|\mv{w}_i\|^2\leq \bar{P}} \bigg\{(R,E):R \leq r_0, E\leq
\sum\limits_{k=1}^KE_k\bigg\},\end{align}where $r_0$ and $E_k$ are given in (\ref{eqn:secrecy rate}) and (\ref{eqn:harvested energy}), respectively. Note that by solving Problem (P1) with each scheme for all feasible $\bar{r}_0$'s, we can characterize the boundary of the corresponding R-E region.

\begin{figure}
\setlength{\abovecaptionskip}{-5pt}
\begin{center}
 \scalebox{0.55}{\includegraphics*{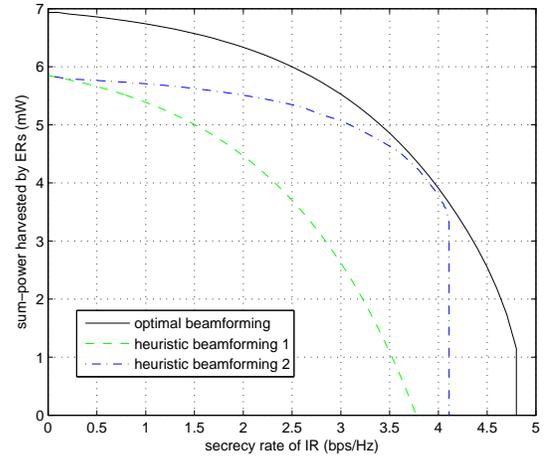}}
 \end{center}
\caption{Achievable R-E regions by the proposed schemes.} \label{fig3} \vspace{-5pt}
\end{figure}

Fig. \ref{fig3} shows a set of R-E regions achieved by different beamforming schemes. It is observed that the optimal beamforming scheme achieves the best R-E trade-off. Moreover, the suboptimal scheme \uppercase\expandafter{\romannumeral2} works better than suboptimal scheme \uppercase\expandafter{\romannumeral1} since its achieved R-E region is closer to that achieved by the optimal scheme, especially when the secrecy rate target for the IR is large. However, it is worth noting that the suboptimal scheme \uppercase\expandafter{\romannumeral1} is of the lowest complexity. Notice that for this suboptimal scheme, the closed-form expressions of the optimal information/energy beamforming vectors and their power allocation are given in (\ref{eqn:sub1 v0}) and (\ref{eqn:sub1 wi}). Furthermore, since no information leakage to ERs is achieved by the designed information beamforming, i.e., $\mv{v}_0^H\mv{g}_k=0$, $\forall k$, there is no need to design any special codebook for the secrecy information signal at Tx.

\section{Conclusion}\label{sec:Conclusion}
This paper is an initial attempt to address the important issue of physical layer security in the emerging simultaneous wireless information and power transfer (SWIPT) system. Under a MISO setup, the joint information and energy beamforming is investigated to maximize the weighted sum-energy harvested by multiple ERs subject to a given secrecy rate constraint at one single IR. We solve this non-convex optimization problem by a two-step algorithm and show that the technique of SDR yields the optimal beamforming solution. Two suboptimal beamforming schemes of lower complexity are also presented, and their performances are compared with that of the optimal scheme in terms of the achievable (secrecy) rate-energy region.

\linespread{0.8}

\end{document}